%% file: root.tex
\documentclass[letterpaper, 10 pt, conference]{ieeeconf}
\IEEEoverridecommandlockouts    
\usepackage{graphics} 
\usepackage{epsfig} 
\usepackage{times} 
\usepackage{amsmath} 
\usepackage{amssymb}  

\usepackage{subfigure}
\usepackage{mathtools}
\usepackage{color}
\usepackage{cite}
\usepackage{comment}
\usepackage{balance}
\usepackage[table]{xcolor}
\usepackage{multirow}
\usepackage{hyperref}
\hypersetup{
    colorlinks=true,
    linkcolor=black,
    citecolor=green,
    urlcolor=blue,
}
\usepackage{blindtext}
\usepackage{bm}
\usepackage{algorithm}
\usepackage{algpseudocode}
\usepackage[official]{eurosym}
\usepackage{listings}
\usepackage{xcolor}
\definecolor{codegreen}{rgb}{0,0.6,0}
\definecolor{codegray}{rgb}{0.5,0.5,0.5}
\definecolor{codepurple}{rgb}{0.58,0,0.82}
\definecolor{backcolour}{rgb}{0.95,0.95,0.92}
\lstdefinestyle{mystyle}{
    backgroundcolor=\color{backcolour},   
    commentstyle=\color{codegreen},
    keywordstyle=\color{magenta},
    numberstyle=\tiny\color{codegray},
    stringstyle=\color{codepurple},
    basicstyle=\ttfamily\footnotesize,
    breakatwhitespace=false,         
    breaklines=true,                 
    captionpos=b,                    
    keepspaces=true,                 
    numbers=left,                    
    numbersep=5pt,                  
    showspaces=false,                
    showstringspaces=false,
    showtabs=false,                  
    tabsize=2
}
\algrenewcommand\algorithmicindent{1.0em}
\algnewcommand\algorithmicswitch{\textbf{switch}}
\algnewcommand\algorithmiccase{\textbf{case}}
\algnewcommand\algorithmicassert{\texttt{assert}}
\algnewcommand\Assert[1]{\State \algorithmicassert(#1)}%
\algdef{SE}[SWITCH]{Switch}{EndSwitch}[1]{\algorithmicswitch\ #1\ \algorithmicdo}{\algorithmicend\ \algorithmicswitch}%
\algdef{SE}[CASE]{Case}{EndCase}[1]{\algorithmiccase\ #1}{\algorithmicend\ \algorithmiccase}%
\algtext*{EndSwitch}%
\algtext*{EndCase}%

\makeatletter
\newcommand{\algmargin}{\the\ALG@thistlm}
\makeatother
\newlength{\forwidth}
\algdef{SE}[parFOR]{parFor}{EndparFor}[1]
  {\parbox[t]{\dimexpr\linewidth-\algmargin}{%
     \hangindent\forwidth\strut\algorithmicfor\ #1\ \algorithmicdo\strut}}{\algorithmicend\ \algorithmicfor}%
\newlength{\forif}
\algnewcommand{\parState}[1]{\State%
  \parbox[t]{\dimexpr\linewidth-\algmargin}{\strut #1\strut}}

\newtheorem{definition}{Definition}[section]
\newtheorem{lemma}{Lemma}[section]

\newtheorem{corollary}{Corollary}[section]
\newtheorem{proposition}{Proposition}[section]
\newtheorem{theorem}{Theorem}[section]
\newtheorem{remark}{Remark}[section]

\newtheorem{problem}{Problem}[section]

\title{\LARGE \bf
Guaranteed Completion of Complex Tasks via \\ Temporal Logic Trees and Hamilton-Jacobi Reachability
}

\author{Frank J. Jiang$^1$, Kaj Munhoz Arfvidsson$^1$, Chong He$^2$, Mo Chen$^2$, Karl H. Johansson$^1$,%
\thanks{
    This work was partially supported by the Wallenberg Artificial Intelligence, Autonomous Systems, and Software Program (WASP) funded by the Knut and Alice Wallenberg Foundation. It was also partially supported by the Swedish Research Council, Swedish Research Council Distinguished Professor Grant 2017-01078, the Knut and Alice Wallenberg Foundation Wallenberg Scholar Grant, and the Swedish Innovation agency (Vinnova), under grant 2021-02555 Future 5G Ride, within the Strategic Vehicle Research and Innovation program (FFI).}
\thanks{
    $^{1}$F. J. Jiang, K. Munhoz Arfvidsson, and K. H. Johansson are with the Division of Decision and Control Systems, EECS, KTH Royal Institute of Technology, Malvinas v{\"a}g 10, 10044 Stockholm, Sweden, email: {\tt\small \{frankji, kajarf, kallej\}@kth.se}. They are also affiliated with the Integrated Transport Research Lab and Digital Futures.}%
\thanks{
    $^{2}$C. He and M. Chen are with the School of Computing Science, Simon Fraser University, Burnaby, BC V5A 1S6, Canada, email: {\tt\small chong\_he@sfu.ca, mochen@cs.sfu.ca}.}%
}

\begin{document}

\maketitle
\thispagestyle{empty}
\pagestyle{empty}

\begin{abstract}
In this paper, we present an approach for guaranteeing the completion of complex tasks with cyber-physical systems (CPS). Specifically, we leverage temporal logic trees constructed using Hamilton-Jacobi reachability analysis to (1) check for the existence of control policies that complete a specified task and (2) develop a computationally-efficient approach to synthesize the full set of control inputs the CPS can implement in real-time to ensure the task is completed. We show that, by checking the approximation directions of each state set in the temporal logic tree, we can check if the temporal logic tree suffers from the ``leaking corner issue,'' where the intersection of reachable sets yields an incorrect approximation. By ensuring a temporal logic tree has no leaking corners, we know the temporal logic tree correctly verifies the existence of control policies that satisfy the specified task. After confirming the existence of control policies, we show that we can leverage the value functions obtained through Hamilton-Jacobi reachability analysis to efficiently compute the set of control inputs the CPS can implement throughout the deployment time horizon to guarantee the completion of the specified task. Finally, we use a newly released Python toolbox to evaluate the presented approach on a simulated driving task.
\end{abstract}

\input{tlt_hj/src/0_macros}
\section{Introduction}\label{sec:intro}
\input{tlt_hj/src/1_intro}

\section{Preliminaries}\label{sec:prelim}
\input{tlt_hj/src/2_prelim}

\section{Challenge: Completion of Complex Tasks}\label{sec:mot}
\input{tlt_hj/src/3_mot}

\section{Constructing TLT using HJ Reachability}\label{sec:method}
\input{tlt_hj/src/4_method}

\section{Computing Least-Restrictive Control Sets}\label{sec:ctrl}
\input{tlt_hj/src/5_ctrl}

\section{Simulated Driving Example}\label{sec:ex}
\input{tlt_hj/src/6_examples}

\section{Conclusion}\label{sec:conc}
\input{tlt_hj/src/7_conclusion}
\balance


\bibliographystyle{IEEEtran}
\bibliography{tlt_hj/tlt_hjb.bib}


\pagebreak

\appendices

\input{tlt_hj/src/8_appendix}
    
\end{document}

%% file: tlt_hj/src/0_macros.tex
\newcommand{\statespace}{\mathbb R^{n_x}}
\newcommand{\inputspace}{\mathbb R^{n_u}}
\newcommand{\inputbound}{\mathcal U}
\newcommand{\ctrlfuncspace}{\mathbb U}

\newcommand{\reach}{\mathcal R}
\newcommand{\rci}{\mathcal{RCI}}

\newcommand{\oreach}{\overline{\reach}}
\newcommand{\ureach}{\underline{\reach}}
\newcommand{\orci}{\overline{\rci}}
\newcommand{\urci}{\underline{\rci}}

\newcommand{\ltltrue}{\textit{true}}
\newcommand{\ltlfalse}{\textit{false}}
\newcommand{\ltlnot}{\neg}
\newcommand{\ltlor}{\lor}
\newcommand{\Ltlor}{\bigvee}
\newcommand{\ltland}{\land}
\newcommand{\Ltland}{\bigwedge}
\newcommand{\ltlimply}{\rightarrow}
\newcommand{\ltlnext}{\bigcirc}
\newcommand{\ltluntil}{\,\mathsf{U}\,}
\newcommand{\ltlalways}{\square}
\newcommand{\ltleventually}{\lozenge}
\newcommand{\ltlsatisfy}{\models}

%% file: tlt_hj/src/1_intro.tex

Over the past few decades, there has been a significant surge in interest towards the development of control techniques for CPS that offer formal safety and liveness guarantees. As CPS become more common in various applications, we need to ensure that these systems not only meet safety or task requirements, but are guaranteed to never violate them during deployment. This challenge has called for the proposal of rigorous methodologies around designing, validating, and implementing controllers for CPS in a way that ensures that safety and liveness is always fulfilled, even in varying or unpredictable environments.

Many of the developed approaches are based on safety/liveness filters or automata-based temporal logic approaches. For CPS, there is a large variety of safety/liveness filter-based control approaches~\cite{Wabersich2023}, such as Hamilton-Jacobi reachability analysis-based approaches~\cite{Mitchell2002, Bansal_Chen_Herbert_Tomlin_2017}, control barrier function-based approaches~\cite{Ames2016, Choi2021}, and zonotope-based approaches~\cite{Althoff2015, Kochdumper2021}. Since safety/liveness filters are developed around the analysis of the propagation of dynamics, the resultant controllers benefit from strong, low-level safety guarantees that take into account phenomenon such as the nonlinearity of the underlying dynamics or bounded disturbances. However, with many safety/liveness filter-based approaches, the tasks that are being solved are usually simple reach-avoid problems and the difficulty of the safety filter design can grow quickly as the complexity of the task grows. In contrast, automata-based temporal logic control approaches, such as~\cite{Yordanov2011, Karaman2008, Wongpiromsarn2012, Ulusoy2014, Belta2017}, leverage the richness of temporal logic languages like linear temporal logic to specify, verify, and synthesize control policies for CPS. Although automata-based temporal logic approaches are powerful for working with complex tasks, the application of these approaches to nonlinear systems or disturbed systems can be impractical due to poor online scalability~\cite{Gao2021b}.

\begin{figure}[t]
    \centering
    \includegraphics[width=0.48\textwidth]{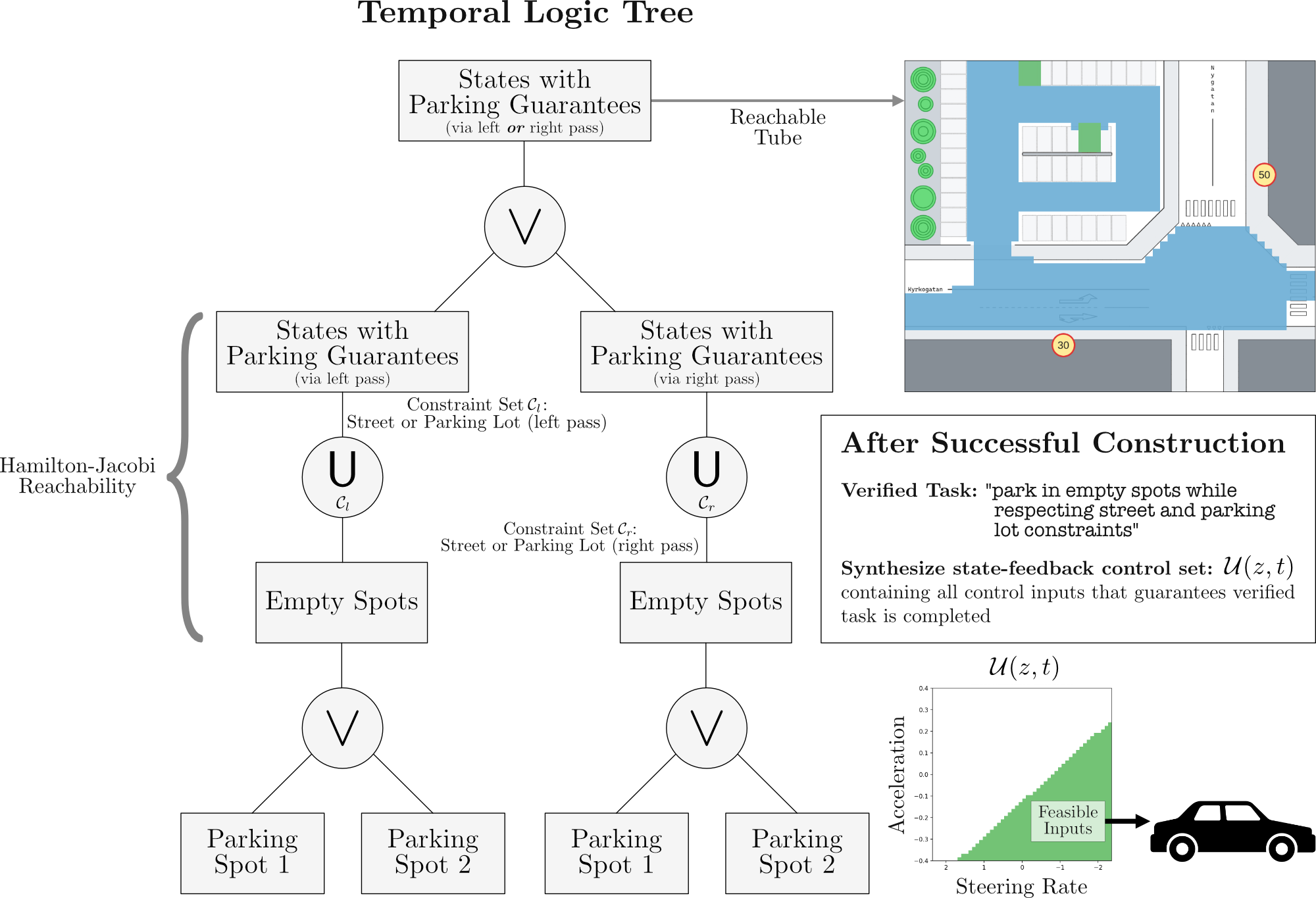}
    \caption{We illustrate and annotate an example temporal logic tree that is used to guarantee the completion of a vehicle parking task.}
    \label{fig:tlt_example}
\end{figure}

To combine the strengths of both approaches, there have been a number of proposals to combine safety/liveness filters with temporal logic over the recent years. In~\cite{Chen2018}, authors explore the use of Hamilton-Jacobi reachability analysis to synthesize control sets for satisfying signal temporal logic specifications. In~\cite{Lindemann2018}, authors explore the application of control barrier functions to efficiently synthesize control policies for a signal temporal logic fragment. More recently, authors introduce a tree-based computation model called temporal logic trees that directly utilizes backward reachability analysis to verify and synthesize control sets for linear temporal logic~\cite{Gao2021b} and signal temporal logic~\cite{Pian2023}. While temporal logic trees have shown initial promise in CPS applications such as automated parking~\cite{Jiang2020a} and remote driving~\cite{Jiang2020b}, there are still a number of challenges with the general application of temporal logic trees. Notably, one of the main challenges is explicitly synthesizing the set of control inputs that is guaranteed to satisfy the constructed temporal logic tree.

\subsection{Contribution}
The main contribution of this paper is an approach that efficiently synthesizes the least-restrictive set of control inputs that a CPS can implement to guarantee the completion of a specified task. To do this, we leverage the value functions resulting from Hamilton-Jacobi reachability analysis to efficiently compute least-restrictive control sets using a computation inspired by the work presented in~\cite{he2023efficient}. We start by detailing how to construct a temporal logic tree using Hamilton-Jacobi reachability analysis and show how we can check if there exist control policies that satisfy the constructed temporal logic tree. Then, we develop a computationally efficient approach to synthesizing least-restrictive control sets from the value functions underlying the constructed temporal logic tree and evaluate the approach on a simulated driving task.
Explicitly, the contributions of this paper can be summarized as follows:

\begin{enumerate}
    \item we detail an algorithm for checking for the existence of control policies that satisfy a constructed temporal logic tree,
    \item we introduce a computationally-efficient approach for explicitly computing least-restrictive control sets from temporal logic trees,
    \item we evaluate the methods presented in this paper on a simulated driving task using the newly open-sourced toolbox called ``Python Specification and Control with Temporal Logic Trees'' (pyspect)\footnote{\url{https://github.com/KTH-SML/pyspect}}.
\end{enumerate}

We also note that while pyspect is exemplified using Hamilton-Jacobi reachability analysis, the toolbox is designed to easily work with any Python-based reachability analysis.

%% file: tlt_hj/src/2_prelim.tex
In this section, we recall and introduce preliminary material that we use in the rest of the paper. Then, in the following section, we start using this material to clearly state the challenges and problems addressed in this work.

\subsection{System Dynamics}
In this work, we consider systems with the following control-affine dynamics
\begin{equation}\label{eq:system}
    \dot z = f(z) + g(z) u,
\end{equation}
where, $z\in\statespace$ and $u\in\inputbound \subset \inputspace$. $f$ and $g$ is uniformly continuous, bounded, and Lipschitz continuous in $z$. Given deployment time horizon $T$, we denote control functions as $u(\cdot): [0, T] \rightarrow \inputbound$, which we assume are measurable, and let $\ctrlfuncspace$ be the function space containing all $u(\cdot)$. Let $\zeta(t; z_0, t_0, u(\cdot)) \in \statespace$ be the state of system~\eqref{eq:system} at time $t$ along a trajectory starting from initial state $\zeta(t_0; z_0, t_0, u(\cdot))=z_0$ under $u(\cdot)$. For simplicity, we will sometimes write $\zeta(\cdot)$ to denote a trajectory of system~\eqref{eq:system}.

\subsection{Temporal Logic}
In this section, we introduce the temporal logic we use to define complex tasks for system~\eqref{eq:system}. In this work, we work with linear temporal logic (LTL). While LTL is a simpler logic compared to other popular logics like signal temporal logic, we choose to work with LTL in this work since we can express sufficiently complex tasks for our examples.

An LTL formula is defined over a finite set of atomic propositions $\mathcal{AP}$ with both logic and temporal operators. We can describe the syntax of LTL with:
\begin{eqnarray*}
\varphi ::= \ltltrue \mid p\in \mathcal{AP} \mid \ltlnot \varphi \mid \varphi_1 \ltland \varphi_2   \mid  \varphi_1 \ltluntil \varphi_2,
\end{eqnarray*}
where $\ltluntil$ denotes the ``until" operators. By using the negation operator and the conjunction operator, we can define disjunction, $\varphi_1\ltlor \varphi_2=\ltlnot (\ltlnot \varphi_1 \ltland \ltlnot \varphi_2)$. Then, by employing the until operator, we can define: (1) eventually, $\ltleventually \varphi=\ltltrue \ltluntil \varphi$ and (2) always, $\ltlalways \varphi= \ltlnot  \ltleventually \ltlnot \varphi$. In this work, we omit the next operator $\bigcirc$, since we develop our approach using Hamilton-Jacobi reachability analysis for continuous time models like~\eqref{eq:system}. Instead of working with LTL over infinite traces, we work with LTL over finite traces. The semantics for LTL over finite traces can be adapted from the semantics of the more common LTL over infinite traces by introducing a ``last'' time $T$ and replacing $\infty$ with $T$~\cite{Giacomo2013}. This is particularly useful for working with general, nonlinear reachability analysis approaches as they typically do not compute or approximate infinite horizon reachable sets. In this work, we refer to $T$ as the ``deployment time horizon''.

\begin{definition}\label{def:semantics}
  \textbf{(LTL semantics)} For an LTL formula $\varphi$, a trajectory $\zeta(\cdot)$, a deployment time horizon $T$, and a time instant check $0 \leq t \leq T$, the satisfaction relation $(\zeta(\cdot),t)\ltlsatisfy \varphi$  is defined as
  \begin{eqnarray*}
  &&(\zeta(\cdot),t) \ltlsatisfy p\in \mathcal{AP} \Leftrightarrow p \in l(\zeta(t)), \\
  && (\zeta(\cdot),t) \ltlsatisfy  \ltlnot  \varphi \Leftrightarrow (\zeta(\cdot),t) \nvDash \varphi, \\
  && (\zeta(\cdot),t) \ltlsatisfy \varphi_1 \ltland \varphi_2 \Leftrightarrow (\zeta(\cdot),t) \ltlsatisfy \varphi_1 \ltland  (\zeta(\cdot),t) \ltlsatisfy \varphi_2, \\
   && (\zeta(\cdot),t) \ltlsatisfy \varphi_1 \ltlor  \varphi_2 \Leftrightarrow (\zeta(\cdot),t) \ltlsatisfy \varphi_1 \ltlor  (\zeta(\cdot),t) \ltlsatisfy \varphi_2, \\
  && (\zeta(\cdot),t) \ltlsatisfy \varphi_1 \ltluntil \varphi_2 \Leftrightarrow \exists t_1\in [t, T] \ \text{s.t.} \\
 &&  \hspace{3.5cm} \begin{cases}
                (\zeta(\cdot),t_1) \ltlsatisfy \varphi_2, \\
                \forall t_2\in [t,t_1), (\zeta(\cdot),t_2) \ltlsatisfy \varphi_1,
              \end{cases} \\
  && (\zeta(\cdot),t) \ltlsatisfy \ltleventually \varphi \Leftrightarrow \exists t_1\in [t, T], \ \text{s.t.} \ (\zeta(\cdot),t_1) \ltlsatisfy \varphi, \\
  && (\zeta(\cdot),t) \ltlsatisfy \ltlalways \varphi \Leftrightarrow \forall t_1\in [t,T], \ \text{s.t.} \ (\zeta(\cdot),t_1) \ltlsatisfy \varphi.
  \end{eqnarray*}
where $p$ is an atomic proposition and $l(\cdot)$ is a labeling function defined as $l: \statespace \rightarrow 2^{\mathcal{AP}}$. For a proposition $p$, we define the following function for relating the proposition to a state set: $\mathcal L^{-1}(p) = \{ z \in \statespace \; | \; p\in l(z)\}$. Using $l(\cdot)$ and $\mathcal L^{-1} (\cdot)$, we are able to associate sets in the state space of system~\eqref{eq:system} with atomic propositions.
\end{definition}

\begin{definition}\label{def:true_sat}
    \textbf{(True Satisfaction Set)} For an LTL formula $\varphi$, we say $\Phi$ is $\varphi$'s true satisfaction set when $\Phi$ is the largest set in $\statespace$ where $\forall z \in \Phi$, $\exists u(\cdot) \in \ctrlfuncspace$, $\forall t \in [0, T]$, $(\zeta(t; z, 0, u(\cdot)), t) \ltlsatisfy \varphi$.
\end{definition}

\subsection{Temporal Logic Trees}
Once we have specified an LTL specification for our system, we can use temporal logic trees (TLT) to check the feasibility and synthesize control sets for satisfying the specification. In Fig.~\ref{fig:tlt_example}, we illustrate an example of a TLT that is constructed for a parking task.

\begin{definition}
  \textbf{(Temporal Logic Tree)} A temporal logic tree (TLT) is a tree for which each node is either a state set node corresponding to a subset of $\mathbb{R}^{n_z}$, or an operator node corresponding to one of the operators $\{\ltlnot, \ltland,\ltlor,\ltluntil,\ltlalways\}$; the root node and the leaf nodes are state set nodes; if a state set node is not a leaf node, its unique child  is an operator node; the children of any operator node are state set nodes.
\end{definition}

From a specified LTL formula, we can follow~\cite[Algorithm]{Gao2021b} to construct a temporal logic tree. Once the temporal logic tree is constructed, we can evaluate whether the LTL formula is satisfiable by checking whether the root of the temporal logic tree is an empty set or not~\cite[Theorem V.1]{Gao2021b}.

\subsubsection{Reachability Analysis}
To construct a TLT, we need to compute the following types of reachable tubes.

\begin{definition}\label{def:brs}
\textbf{(Backward Reachable Tube)}
Given system~\eqref{eq:system}, a target set $\mathcal T \subseteq \statespace$, and a constraint set $\mathcal C \subseteq \statespace$, we define the backward reachable tube as
    \[\begin{split}
    \mathcal R(\mathcal T; \mathcal C) = \{ 
        z \; | \;
            & \exists u(\cdot) \in \ctrlfuncspace,\\
            & \exists \tau \in [0,T), \: \zeta(\tau; z, 0, u(\cdot)) \in \mathcal T, \\
            & \forall \tau' \in [0, \tau), \:
            \zeta(\tau'; z, 0, u(\cdot)) \in \mathcal C
    \},
    \end{split}\]
where $\mathcal R(\mathcal T; \mathcal C)$ contains the set of states that are able to reach the target set $\mathcal T$ while staying within constraint set $\mathcal C$. For simplicity, we will denote this operation with $\mathcal R(\cdot)$.
\end{definition}

\begin{definition}\label{def:rcis}
\textbf{(Robust Control Invariant Set)} For system~\eqref{eq:system} and constraint set $\mathcal C \subseteq \statespace$, $\mathcal{RCI}(\mathcal C) \subseteq \statespace$ the largest robust control invariant set such that $\forall z \in \mathcal{RCI}(\mathcal C)$ there $\exists u(\cdot)\in \ctrlfuncspace$ such that $\forall \tau \in [0, T]$, $\zeta(\tau; z, 0, u(\cdot)) \in \mathcal C$.
\end{definition}

We denote the over- and under-approximation of the true backward reachable tube (or largest RCIS) with $\oreach$ and $\ureach$ (or $\orci$ and $\urci$), respectively.

\subsection{Least-Restrictive Control Sets}
After constructing TLTs, we can compute control sets that can be used to saturate the inputs to system~\eqref{eq:system} to guarantee that the verified task is completed. We define a least-restrictive control set as the following: 

\begin{definition}
    \textbf{(Least-Restrictive Control Set)} Given system~\eqref{eq:system}, state $z \in \statespace$, time $t \in [0, T]$, a reachable tube $S$, the least-restrictive control set is defined as
    \begin{multline}\label{eq:least_restrictive}
        \mathcal U(z, t) = \{u(t): u(\cdot) \in \ctrlfuncspace, \forall \tau \in [t, T], \\ \zeta(\tau; z, t, u(\cdot)) \in \mathcal S\},
    \end{multline}
    where $\mathcal U(z, t)$ is a state and time feedback control set that contains all the control inputs system~\eqref{eq:system} is allowed to implement at state $z$ and time $t$ to stay within the reachable tube for the rest of the deployment time horizon.
\end{definition}

\begin{figure*}
    \centering
    \includegraphics[width=0.98\textwidth]{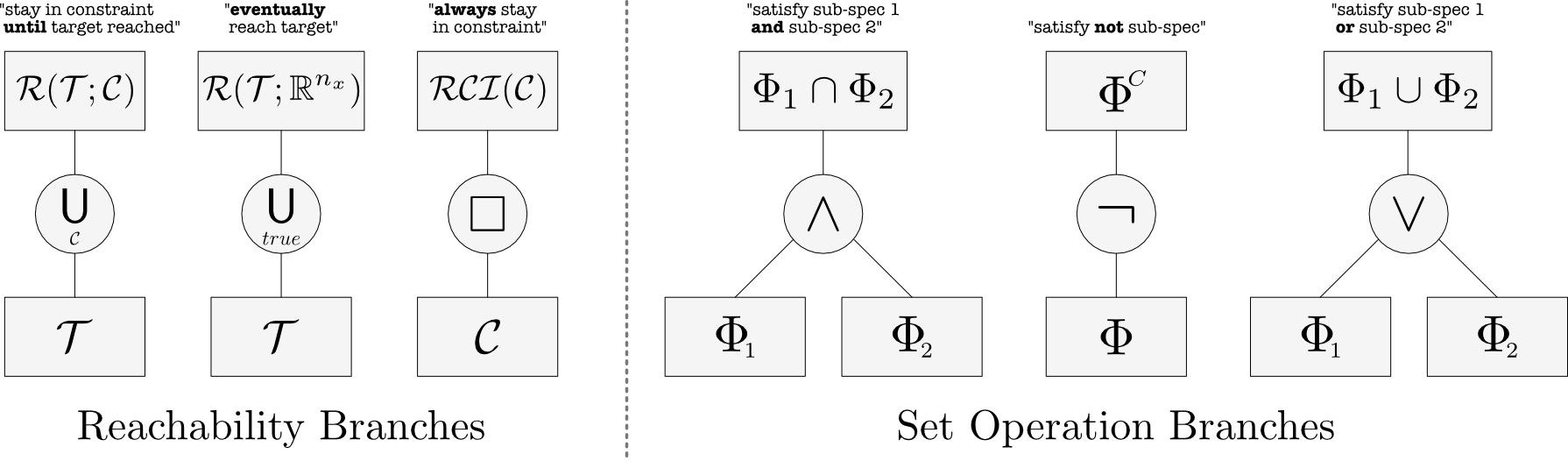}
    \caption{We illustrate and describe all the different branches that can occur in a TLT and their corresponding state set and operator nodes.}
    \label{fig:branches}
\end{figure*}

%% file: tlt_hj/src/3_mot.tex
In this section, we clarify the specific problems we address in this work. To do this, we start by considering an example that will be evaluated at the end of this paper: vehicle parking. As is thoroughly discussed in~\cite{Jiang2020a}, although vehicle parking may seem like a task that is simple, the complexity of the task easily grows in parking scenarios with stringent requirements. There may be varying speed requirements, hard-to-navigate driving spaces, multiple available spots, etc. As the complexity of the task grows or even changes in real-time (in the case where the parking environment changes), it's critically important we are able to guarantee the vehicle is able to safely and successfully complete it's parking task. One approach to this problem is constructing TLTs using Hamilton-Jacobi (HJ) reachability analysis~\cite{Jiang2020a, Jiang2020b, Yu2023}. A challenge faced by these works is that the least-restrictive control set is used implicitly and is not explicitly computed. Additionally, there is still no formal treatment for the challenge of checking whether there exists control policies that satisfy the constructed temporal logic tree. To address these two challenges, we solve the following problems.

\begin{problem}
    Given system~\eqref{eq:system} and a TLT $\mathfrak T$ constructed using HJ reachability analysis for LTL task $\varphi$, guarantee that there exists control policies where system~\eqref{eq:system}'s resultant trajectory fully satisfies $\mathfrak T$.
\end{problem}

\begin{problem}
    Given that there exists control policies such that system~\eqref{eq:system} is able to satisfy the constructed $\mathfrak T$, efficiently utilize the value functions from the HJ reachability analysis to compute a least-restrictive control set $\hat{\mathcal U}$ that contains all of the control inputs system~\eqref{eq:system} can implement in real-time to guarantee it satisfies $\varphi$.
\end{problem}

%% file: tlt_hj/src/4_method.tex
In this section, we will detail the HJ reachability computations necessary for constructing the temporal logic tree. In particular, we introduce the computations involved with approximating the satisfaction sets of $\{\ltluntil, \ltleventually, \ltlalways\}$.

\subsection{Key HJ PDEs}\label{subsec:hj_reach}
We define two key HJ PDEs that need to be solved to approximate solutions to $\mathcal R (\cdot)$ and $\mathcal{RCI}(\cdot)$. For the first PDE, with implicit target surface function $V_{\mathcal T}$ as our initial condition, we solve for solution $V_{\mathcal R}$ under implicit constraint surface function $V_{\mathcal C}$ in the constrained, reach HJ PDE 
\begin{equation}\begin{split}\label{eq:reach_pde}
    \max\{D_{t'} V_{\mathcal R}(z, t') + \min_{u\in\inputbound}D_z V_{\mathcal R}(z, t') \cdot f(z, u),\\ -V_{\mathcal C}(z) - V_{\mathcal R}(z, t')\} &= 0, \\
    V_{\mathcal R}(z, 0) &= V_{\mathcal T}(z), \\
\end{split}\end{equation}
where $t' \leq 0$ is the time used to solve the PDE backwards in time. With $V_{\mathcal C}$ and $V_{\mathcal T}$ defined such that $\mathcal C$ and $\mathcal T$ are their respective zero sub-level sets (e.g. using a signed distance function), we compute the backward reachable tube defined in Definition~\ref{def:brs} as the zero sub-level set of the solution of~\eqref{eq:reach_pde} at time $t' = -T$, which we write explicitly as
\begin{equation}\label{eq:brs_computation}
\mathcal R(\mathcal T; \mathcal C) = \{z \; | \; V_{\mathcal R}(z, -T) \leq 0\}.
\end{equation}
Then, for the second PDE, it will be useful to introduce the concept of avoid tubes, although they are not always necessary for temporal logic trees.
\begin{definition}\label{def:avoid}
    \textbf{(Avoid Tube)} Given the system~\eqref{eq:system} and a set the system should stay outside of $\mathcal O \subset \statespace$, we define an avoid tube as
    \[\begin{split}
        \mathcal A(\mathcal O) = \{ 
            z \; | \; \forall u(\cdot) \in \ctrlfuncspace,
            \exists \tau \in [0,T), \: \zeta(\tau; z, 0, u(\cdot)) \in \mathcal A
        \},
    \end{split}\]
\end{definition}

\smallskip

We then introduce the second HJ PDE that we solve to find avoid tubes. With the implicit avoid surface function $V_{\mathcal O}$ as an initial condition, we solve for $V_{\mathcal A}$ in the following avoid HJ PDE,
\begin{equation}\begin{split}\label{eq:avoid_pde}
    D_{t'} V_{\mathcal A}(z, t') + \max_{u\in\inputbound}D_z V_{\mathcal A}(z, t') \cdot f(z, u) &= 0, \\
    V_{\mathcal A}(z, 0) &= V_{\mathcal O}(z), \\
\end{split}\end{equation}
We can compute the avoid tube defined in Definition~\ref{def:avoid} as the zero sub-level set of the solution of~\eqref{eq:avoid_pde} at time $t' = -T$, which we write explicitly as
\begin{equation}\label{eq:avoid_computation}
    \mathcal A(\mathcal O) = \{z \; | \; V_{\mathcal A}(z, -T) \leq 0\}.
\end{equation}
Next, we show how we compute~\eqref{eq:brs_computation} and~\eqref{eq:avoid_computation} to approximate the satisfaction sets for $\{\ltluntil, \ltleventually, \ltlalways\}$ and verify that there exists control policies that satisfy each operator.

\subsection{Approximating $\ltluntil$ (Until) and $\ltleventually$ (Eventually)}

Let $\varphi_1$ and $\varphi_2$ be two LTL sub-formulae. In the left-most branch in Fig.~\ref{fig:branches}, we can see how $\varphi_1 \ltluntil \varphi_2$ looks as a TLT branch when $\mathcal C$ and $\mathcal T$ correspond to the satisfaction sets of $\varphi_1$ and $\varphi_2$, respectively. Now, let $\Phi_1$, $\Phi_2$, $V_{\varphi_1}$, $V_{\varphi_2}$ be the satisfaction state sets and surface functions corresponding to LTL formulae $\varphi_1$ and $\varphi_2$. Then, we approximate $\varphi_1 \ltluntil \varphi_2$ based on~\eqref{eq:brs_computation} by finding
\begin{equation}\label{eq:until_computation}
    \mathcal R(\Phi_2; \Phi_1) = \{z \; | \; V_{\mathcal R}(z, -T) \leq 0\}.
\end{equation}
In other words, the until operator corresponds to solving a constrained backward reachability problem with the target as the satisfaction set of $\varphi_2$ and with the constraint as the satisfaction set of $\varphi_1$. Since $\ltleventually \varphi_2 = true \ltluntil \varphi_2$, we perform a similar computation to~\eqref{eq:until_computation} to find the satisfaction set for the eventually operator. However, since $\Phi_1$ becomes the satisfaction set of $true$, we end up solving an unconstrained backward reachability problem, as is illustrated in the middle reachability branch in Fig.~\ref{fig:branches}. For most solvers, this is equivalent to dropping $V_{\mathcal C}(z)$ in~\eqref{eq:reach_pde}.
Depending on how many ancestor $\ltlnot$ operator nodes the $\ltluntil$ node has, we may need to change whether we under- ($\ureach(\cdot)$) or over-approximate ($\oreach(\cdot)$) solutions to~\eqref{eq:brs_computation}. This can be done numerically in level-set methods for solving~\eqref{eq:brs_computation}~\cite{Mitchell2002}.

\subsection{Approximating $\ltlalways$ (Always)}

Let $\varphi$ and $\Phi \subset \statespace$ be an LTL sub-formula and the corresponding satisfaction set. In Definitions~\ref{def:rcis}, $\mathcal{RCI}(\cdot)$ is defined around the existence of a control policy $u(\cdot)$ that keeps system~\eqref{eq:system} within a set for all time. Since this is consistent with the satisfaction requirement of $\ltlalways \varphi$, to approximate the satisfaction set of $\ltlalways \varphi$ we can find an approximate solution to $\mathcal {RCI}(\Phi)$.  In Fig.~\ref{fig:branches}, we can see how $\ltlalways \varphi$ looks as a TLT branch when $\mathcal C = \Phi$. For approximating $\mathcal{RCI}(\Phi)$, we utilize the avoid tube defined in Definition~\ref{def:avoid}. Instead of directly approximating the largest set of states where there exists control policies that keep the system in $\Phi$, it is more common to approximate the set of states where for all control policies force the system outside of $\Phi$ and taking the complement of this set. In other words, we compute $\mathcal{RCI}(\cdot)$ as the following:
\begin{equation}\label{eq:rci}
    \mathcal{RCI}(\Phi) = \mathcal A^C(\Phi^C).
\end{equation}
Based on~\eqref{eq:avoid_computation} and with $\mathcal O = \Phi^C$, we approximate $\ltlalways \varphi$ with
\begin{equation}
    \mathcal{RCI}(\Phi) = \{z \; | \; V_{\mathcal A}(z, -T) > 0\}.
\end{equation}
Similar to the $\ltluntil$ operator, depending on how many ancestor $\ltlnot$ operator nodes the $\ltlalways$ node has, we may need to change whether we under- ($\urci(\cdot)$) or over-approximate ($\orci(\cdot)$) solutions, which is done by numerically over- or under-approximating ~\eqref{eq:avoid_computation}, respectively.

%% file: tlt_hj/src/5_ctrl.tex

\begin{algorithm}[t]
\small
\caption{\texttt{ctrlExists}}
\label{alg:approx}
\textbf{Input}: Root node $\mathfrak R$ of a constructed TLT \\
\textbf{Output}: \texttt{E}, \texttt{O}, \texttt{U}, or \texttt{I}

\begin{algorithmic}[1]
    \If{\texttt{isLeaf}($\mathfrak R$)}\label{ln:leaf}
        \State \Return \texttt{approxDirection}($\mathfrak R$) \textit{\#} \texttt{E}, \texttt{O}, \texttt{U}
    \EndIf
    \State $c$ = \texttt{child}($\mathfrak R$)
    \State $G$ = \texttt{grandChildren}($\mathfrak R$)
    \If{\texttt{length}($G$) == 1} 
        \State $g$ = $G$
        \State $a$ = \texttt{ctrlExists}($g$)
        \State \textbf{if} {$a$ == \texttt{I}} \textbf{then} \Return \texttt{I}\label{ln:inv_unary}
        \State \textit{\# check operator nodes with single child}
        \If{$c$ == $\ltluntil$ or $c$ == $\ltlalways$}\label{ln:until_alw}
            \State $a_o$ = \texttt{approxDirection}($c$)
            \State \textbf{if} $a_0$ != $a$ \textbf{then} \Return \texttt I \textbf{else} \Return $a$
        \ElsIf{$c$ == $\ltlnot$}\label{ln:not}
            \State \Return -1 * $a$
        \EndIf
    \Else
        \State $g_1$, $g_2$ = $G$
        \State $a_1$, $a_2$ = \texttt{ctrlExists}($g_1$), \texttt{ctrlExists}($g_2$)
        \State \textbf{if} $a_1$ == \texttt{I} or $a_2$ == \texttt{I} \textbf{then} \Return \texttt{I}\label{ln:inv_binary}
        \State \textit{\# check operator nodes with two children}
        \If{$c$ == $\ltland$}\label{ln:and}
            \State \textbf{if} $a_1$==\texttt{U} or $a_2$==\texttt{U} \textbf{then}
                \Return \texttt{I} \textbf{else} \Return \texttt O
        \ElsIf{$c$ == $\ltlor$}\label{ln:or}
            \If {$a_1$ == $a_2$}
                \Return $a_1$
            \ElsIf {$a_1$ == \texttt{E}}
                \Return $a_2$
            \ElsIf {$a_2$ == \texttt{E}}
                \Return $a_1$
            \Else {}
                \Return \texttt{I}
            \EndIf
        \EndIf 
    \EndIf
\end{algorithmic}
\end{algorithm}

In this section, we present our approach for efficiently computing least-restrictive control sets from a TLT constructed using HJ reachability analysis. We start by introducing an algorithm for checking quickly checking if control policies still exist for satisfying the TLT. Then, we introduce the explicit computation of the least-restrictive control sets that enable a system to follow the existing control policies.

\subsection{Checking for existing control policies}\label{subsec:tlt}

From~\cite{Gao2021b}, we know that the constructed TLT for an LTL specification $\varphi$ should under-approximate the true satisfaction set of $\varphi$. However, as is emphasized in~\cite[Theorem V.1]{Gao2021b} the under-approximation is tied to the existence of control policies that satisfy the constructed TLT. In some applications or tasks, this can be obvious to the designer. However, in many applications or tasks, this should be automatically checked. One of the key challenges with automatically checking and ensuring that there exists control policies that satisfy the constructed TLT is the so-called ``leaking corner problem''. The leaking corner problem occurs when two satisfaction sets are intersected. This problem also arises in the context of system decomposition~\cite{Chen2018a, he2023efficient} and classic reach-avoid problems~\cite{Lee2019}. In the construction of TLTs, this problem is further complicated by the fact that in an LTL formula, there may be many nested $\{\ltlnot, \ltland, \ltlor\}$, which all have requirements and effects around approximation directions. To address this, we start by characterizing the effect of the set operations underlying $\{\ltlnot, \ltland, \ltlor\}$.

\begin{lemma}\label{lemma:not}
\textbf{(Effect of $\ltlnot$)} Let $\varphi$ be an LTL subformula, $\Phi \subset \statespace$ be $\varphi$'s true satisfaction set. Similarly, let $\Phi'$ be the true satisfaction set of $\ltlnot \varphi$. Also, let $\hat \Phi$ be either an over- or under-approximation of $\Phi$.
\begin{itemize}
    \item Let $\hat \Phi$ be an over-approximation of $\Phi$ ($\hat \Phi \supset \Phi$), then $\hat \Phi^C$ is an under-approximation of $\Phi'$ ($\hat \Phi^C \subset \Phi'$).
    \item Let $\hat \Phi$ be an under-approximation of $\Phi$ ($\hat \Phi \subset \Phi$), then $\hat \Phi^C$ is an over-approximation of $\Phi'$ ($\hat \Phi^C \supset \Phi'$).
\end{itemize}
\end{lemma}

\begin{proof}
   Proof in Appendix~\ref{app:ctrl_exists}
\end{proof}
In other words, the $\ltlnot$ operator reverses the approximation direction.

\begin{algorithm}[t]
\small
\caption{\texttt{leastRestrictiveCtrl}}
\label{alg:ctrl}
\textbf{Input}: state $z\in\statespace$, time $t$, and constructed TLT $\mathfrak T$ \\
\textbf{Output}: least-restrictive control set $\mathcal U_{z,t} \subseteq \inputbound$ for $z$ and $t$
\begin{algorithmic}[1]
    \State $C$ = \texttt{controlTree}($z$, $t$, $\mathfrak T$) \textit{\#~\cite[Algorithm 4]{Gao2021b}}
    \State $\hat{C}$ = \texttt{compressTree}($C$) \textit{\#~\cite[Algorithm 2]{Gao2021b}}
    \State $\mathcal U_{z,t}$ = \texttt{setBacktrack}($\hat{C}$) \textit{\#~\cite[Algorithm 5]{Gao2021b}}
    \State \Return $\mathcal U_{z,t}$
\end{algorithmic}
\end{algorithm}

\begin{lemma}\label{lemma:and}
\textbf{(Effect of $\ltland$)} Let $\varphi = \varphi_1 \ltland \varphi_2$, and $\Phi, \Phi_1, \Phi_2$ be their respective, true satisfaction sets. Then,
\begin{equation}
    \Phi \subseteq \Phi_1 \cap \Phi_2
\end{equation}
\end{lemma}
\smallskip
\begin{proof}
   Proof in Appendix~\ref{app:ctrl_exists}
\end{proof}
In other words, the set intersection underlying $\ltland$ can lead to over-approximation of the true satisfaction set. This is the key leaking corner issue, as we want the constructed TLT to under-approximate the true satisfaction set. However, this issue can be remedied if the $\ltland$ operator node has an $\ltlnot$ operator node as an ancestor.

\begin{corollary}\label{corr:not}
\textbf{(Effect of $\ltlnot$ on $\ltland$)} Let $\varphi = \varphi_1 \ltland \varphi_2$, and $\Phi, \Phi_1, \Phi_2$ be their respective, true satisfaction sets. Then, we get the following under-approximation of $\ltlnot \varphi$
\begin{equation}
    \Phi^C \supseteq (\Phi_1 \cap \Phi_2)^C
\end{equation}
\end{corollary}
\begin{proof}
    Seen through Lemma~\ref{lemma:not} and Lemma~\ref{lemma:and}.
\end{proof}
In other words, if an $\ltlnot$ is applied to an $\ltland$, the result is an under-approximation. This means that the resulting set only contains states where there exist control policies that satisfy the corresponding sub-formula.

\begin{remark}
Within the scope of this work, we will propose an algorithm for just checking if the leaking corner issue affects the constructed TLT. However, there are a couple other approaches to treat this problem: (1) additional assumptions are made to ensure $\ltland$ results in an exact result, and (2) post-processing of the $\ltland$ operator's underlying set intersection to ensure result is exact or an under-approximation. In some cases, the first approach is viable and an assumption can be made (i.e. one child is the subset of the other child) where leaking corner problems is avoided. However, for cases when the leaking corner problem cannot be avoided, an important future work will be to develop and incorporate new, computationally-efficient approaches to resolve leaking corners for the second approach.
\end{remark}

\begin{lemma}\label{lemma:or}
\textbf{(Effect of $\ltlor$)} Let $\varphi = \varphi_1 \ltlor \varphi_2$, and $\Phi, \Phi_1, \Phi_2$ be their respective, true satisfaction sets. Then,
\begin{equation}
    \Phi \equiv \Phi_1 \cup \Phi_2
\end{equation}
\end{lemma}
\smallskip
\begin{proof}
    Proof in Appendix~\ref{app:ctrl_exists}
\end{proof}
This means that there is no contribution of additional approximation from $\ltlor$ operators when taking the union of the underlying satisfaction sets. That said, there are requirements on the children of $\ltlor$, which are outlined in Algorithm~\ref{alg:approx}.

\begin{figure}[t]
    \centering
    \includegraphics[width=0.48\textwidth]{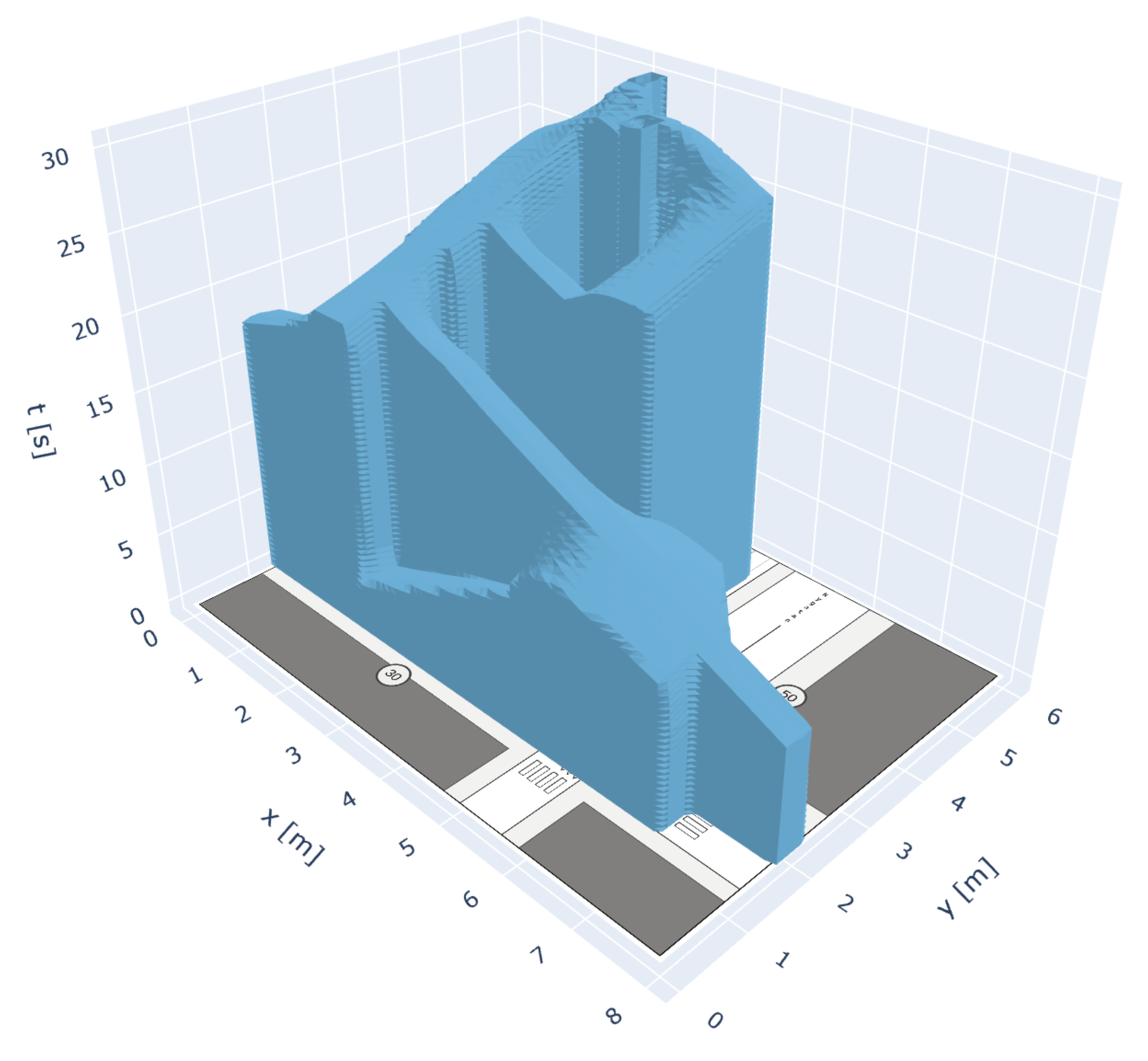}
    \caption{The full, under-approximating satisfaction set for the parking task computed by constructing the temporal logic tree using HJ reachability analysis}
    \label{fig:reach_tube}
\end{figure}

Now that we have characterized the individual effects of the $\{\ltlnot, \ltland, \ltlor\}$ operators on the existence of control policies, we present Algorithm~\ref{alg:approx}. In Algorithm~\ref{alg:approx}, we recurse through the TLT and consider the numerical approximations underlying $\{\ltluntil, \ltlalways\}$ and the approximation effects of $\{\ltlnot, \ltland, \ltlor\}$ to decide whether control policies exist the satisfy the full TLT. For the output of the algorithm, we enumerate the variables \texttt{E} $= 0$, \texttt{O} $= +1$, \texttt{U} $= -1$, and \texttt{I} $=$ \texttt{NaN} corresponding to exact, over-, under-approximation, and invalid, respectively.

\begin{theorem}\label{thm:approx}
    Let $\varphi$ be an LTL formula and $\mathfrak R$ be the non-empty root node of a constructed TLT $\mathfrak T$ for verifying $\varphi$. If \texttt{ctrlExists}($\mathfrak R$) returns \texttt{U}, then there exists control policies that satisfy $\mathfrak T$.
\end{theorem}
\begin{proof}
    Proof in Appendix~\ref{app:ctrl_exists}.
\end{proof}

\subsection{Synthesizing Least-Restrictive Control Sets}
Now, to synthesize the least-restrictive control set for full TLT, we follow~\cite[Algorithm 3]{Gao2021b}. For clarity, we include an adapted version in Algorithm~\ref{alg:ctrl}. We start by computing the individual least-restrictive control sets for $\ltluntil$ and $\ltlalways$ operator nodes (Line 11 and 23 in~\cite[Algorithm 4]{Gao2021b}). Using the value function $V_{\Phi}$ that from performing HJ reachability analysis, this is done by computing the following control set:
\begin{multline}\label{eq:hj_least_restrictive}
        \mathcal{U}_{\Phi}(z, t)=\{u\in \inputbound \; | \; D_t V_{\Phi}(z, -T+t)  + \\ D_z V_{\Phi}(z, -T+t) ^\top (f(z) + g(z) u)\leq 0\}.
\end{multline}
For $0\le t \le T$, let the computational time step be $\delta t$, and  $s=-T+t$. The value function can be modified with first-order Taylor expansion:
\begin{multline}
V_{\Phi}(\zeta(s+\delta t; z, s), s+\delta t) = V_\Phi(z, s)+\\
D_tV_\Phi(z,s)\delta t +D_z V_\Phi(z,s)^\top (f(z)+g(z)u)\delta t \le 0.
\end{multline}
The above equation implies that the satisfaction of $\Phi$ induces the least restrictive control set as the half-space:
\begin{equation}\label{eq:lrctrlbound}
a + b^\top u \le 0
\end{equation}
with
\begin{align*}
&a=V_\Phi(z,s)+D_tV_\Phi(z,s)\delta t+D_zV_\Phi(z,s)^\top f(z) \delta t \\
&b^\top =D_zV_\Phi(z,s)^\top g(z)\delta t
\end{align*}
This results in a least-restrictive control set that is an additional control constraint on top of any original control constraints, such as control bounds. Readers can find a code snippet for computing this control set using the odp toolbox~\cite{Bui2022} in Appendix~\ref{app:lcrs_comp}.

\begin{figure*}
    \centering
    \includegraphics[width=0.98\textwidth]{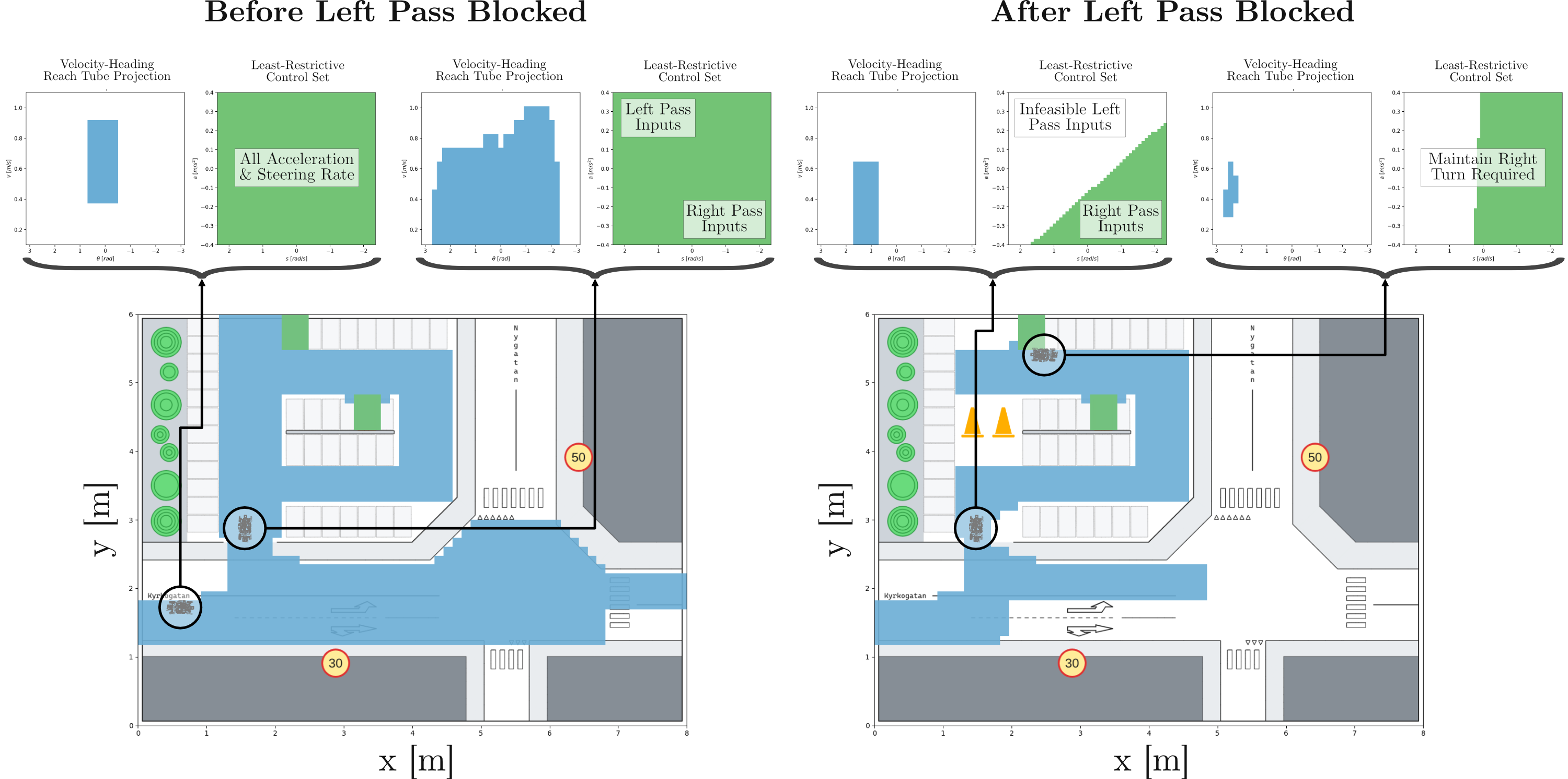}
    \caption{We illustrate the deployment of the simulated vehicle for the specified parking task. In the two larger plots, we show the full environment and an xy projection of the 5D satisfaction state set of the constructed TLT. In the smaller plots, we show the velocity-heading projection of the constructed TLT's satisfaction set (in blue) and the computed least-restrictive control set (in green) for different states and times throughout the task completion.}
    \label{fig:simulation}
\end{figure*}

\begin{proposition}\label{prop:temporal_fragment}
Satisfaction of $\Phi$ is guaranteed if $V_\Phi(z,s)\le0$. In this case, the state $z$ at time $s$ has a nonempty feasible control set. 
\end{proposition}

\begin{proof}
    This is a well-known result (for example, see \cite{Fisac_Chen_Tomlin_Sastry_2015}), and in the ``reach" case follows immediately from Eq. \eqref{eq:reach_pde}.
    The first argument of the $\max$ operator, $D_{t'} V_{\mathcal R}(z, t') + \min_{u\in\inputbound}D_z V_{\mathcal R}(z, t') \cdot f(z, u)$, is the total time derivative of $V_{\mathcal R}$ when the optimal control is applied. 
    Thus, Eq. \eqref{eq:reach_pde} implies that $V_{\mathcal R}$ is non-increasing along optimal trajectories.
    It is also known that $V_{\mathcal R}$ exists and is unique \cite{Barron1989, Barron1990}, so by construction, if $V_{\mathcal R}(z,s)\le0$, there must exist a control to keep the value of $V_{\mathcal R}$ non-positive.
    This can also be shown in a similar argument involving $V_{\mathcal A}$ in the ``avoid case", following Eq. \eqref{eq:avoid_pde}.
    We conclude the proof by noting that $V_\Phi(z,s)$ is equal to either $V_{\mathcal R}$ or $V_{\mathcal A}$ depending on $\Phi$.
\end{proof}

Finally, we can show that the output of Algorithm~\ref{alg:ctrl} will result in the guaranteed completion of specified task.
\begin{theorem}\label{thm:ctrl}
    Let $\varphi$ be an LTL formula and $\mathfrak T$ be the constructed TLT that passes \texttt{ctrlExists}. If at $t=0$, $z(0)$ is in the root node's state set and $\forall t \in [0, T]$, system~\eqref{eq:system} implements $u(t) \in \mathcal U_{z(t), t}$, where $\mathcal U_{z(t), t}$ is the output of \texttt{leastRestrictiveCtrl}($z(t), t, \mathfrak T$), then system~\eqref{eq:system} is guaranteed to satisfy $\varphi$.
\end{theorem}

\begin{proof}
    Since the system starts at time $t=0$ with $z(0)$ in $\mathfrak T$'s root node's state set, then we know the system initially satisfies $\varphi$ since the root node's state set under-approximates the true satisfaction set of $\varphi$. In Algorithm~\ref{alg:ctrl}, when the control tree is synthesized, each state set node in $\mathfrak T$ is replaced by least-restrictive control sets~\eqref{eq:hj_least_restrictive} or the union/intersection of control sets~\eqref{eq:hj_least_restrictive}. We know from Proposition~\ref{prop:temporal_fragment} that the individual least-restrictive control sets that are the parents of $\ltluntil$ and $\ltlalways$ operator nodes in the control tree are nonempty feasible control sets and ensure the system satisfies the corresponding LTL sub-formula. Since we know $\mathfrak T$ passes \texttt{ctrlExists}, we know that $\mathfrak T$ does not have leaking corners and the combined least-restrictive control sets will also be nonempty feasible control sets. Since \texttt{compressTree} and \texttt{setBacktrack} only apply unions to the least-restrictive control sets up $\mathfrak T$, the final output $\mathcal U_{z,t}$ of \texttt{setBacktrack} is also the nonempty feasible control set containing all of the control inputs the system can implement to satisfy $\varphi$ at state $z$ and time $t$.
\end{proof}

Now that we are able to check the existence of satisfying control policies and synthesize least-restrictive control sets for a specified task, we illustrate and evaluate our approach on a simulated driving task in the next section.

%% file: tlt_hj/src/6_examples.tex

In this section, we apply the presented approach to a simulated driving task where a vehicle is tasked to park into a parking lot on a road network inspired by the Kyrkogatan-Nygatan intersection in Eskilstuna, Sweden. For this task, the vehicle starts from somewhere in the road network nearby the parking lot and while it's entering the parking lot, it's task will be changed due to an unplanned blocking of one part of the parking lot. With this example, we illustrate the resultant least-restrictive control sets that are computed during the deployment to show case how the sets correctly constrain the vehicle's acceleration and steering rates to guarantee the parking task is completed. For this example, we utilize the newly released pyspect toolbox, where the code for the example itself can also be found.

\subsection{Nonlinear, 5-state Bicycle Model}
For the vehicle, let $z = [x, y, \theta, \delta, v]^\top$ be the state, where $x$, $y$, $\theta$, $\delta$, and $v$ are the vehicle's x-position, y-position, heading angle, steering angle, and velocity, respectively. Then, let $u = [s, a]^\top$ be the input, where $s$ and $a$ are the steering rate and acceleration inputs into the vehicle, respectively. Explicitly, we write the dynamics as the following:
\begin{equation}
    f(z) =
    \left[\begin{matrix}
        v \cos\theta \\
        v \sin\theta \\
        \frac{v \tan \delta}{L}\\
        0 \\
        0
    \end{matrix}\right], \ 
    g(z) = 
    \left[\begin{matrix}
        0 & 0\\
        0 & 0\\
        0 & 0\\
        1 & 0 \\
        0 & 1
    \end{matrix}\right],
    \nonumber
\end{equation}
where $L$ is the wheel-base length of the vehicle.

\subsection{Task Specification}
When constructing the task specification, we start by defining atomic propositions that describe this environment in terms of road geometries and other state constraints. For convenience, we denote state inequalities using propositions like $x_{<c}$ where $L'(x_{<c}) \equiv \{ z \in \statespace \mid x < c\}$. Then, we define sub-formulae such as $\textit{Kyrkog. Geometry}$, $\textit{P-Lot Geometry}$, etc. that describe the drivable space. We also embed any speed limits and heading constraints into these formulae, for instance,
\begin{equation*}\begin{split}
    \textit{Kyrkogatan} =\;&  \textit{Kyrkog. Geometry}  \\
        & \ltland (y_{< 3}      \ltlimply v_{> 0.4} \ltland v_{<1.0}) \\ 
        & \ltland (y_{\ge 3}    \ltlimply v_{> 0.3} \ltland v_{<0.6}) \\
        & \ltland (y_{< 1.8}    \ltlimply \theta_{> \frac{+\pi}{5}} \ltland \theta_{< \frac{-\pi}{5}}) \\
        & \ltland (y_{\ge 1.8}  \ltlimply \theta_{> \frac{+4\pi}{5}} \ltlor \theta_{< \frac{-4\pi}{5}}).
\end{split}\end{equation*}
For the exact definitions and values, we refer readers to the implementation available in pyspect.

The task is for the vehicle to park in one of the two empty parking spots. Namely, the target is $\textit{Empty Spots} = \textit{Parking Spot 1} \ltlor \textit{Parking Spot 2}$. In practice, picking between these can relate to which is closest or is most accessible. Furthermore, the sub-formulae 
\begin{equation*}\begin{split}
    \varphi_L &= (\textit{Kyrkogatan} \ltlor \textit{P-Lot Left}) \ltluntil \textit{Empty Spots} \quad \text{and} \\
    \varphi_R &= (\textit{Kyrkogatan} \ltlor \textit{P-Lot Right}) \ltluntil \textit{Empty Spots}
\end{split}\end{equation*}
allow the vehicle to go either left or right inside parking lot. With these, the final task specification is $\varphi = \varphi_L \ltlor \varphi_R$ and the corresponding TLT is shown in Fig.~\ref{fig:tlt_example}.

\subsection{Results}

The TLT is evaluated with a time horizon $T = 30\text{ seconds}$. In Fig.~\ref{fig:reach_tube}, we visualize the time evolution of the satisfaction set along the xy-plane. Since $\varphi$ allows the vehicle to reach the empty spots via either left or right pass, we can see the tube filling the entirety of the parking lot's inside. The vehicle can enter the parking lot from both sides of Kyrkogatan and, as such, we can see the tube's ridge near the entrance splitting in both directions. Consider now that the left pass is blocked, see Fig.~\ref{fig:simulation}, which forces the vehicle to go via the right pass. The unplanned block results in the pruning of the TLT branch (shown in Fig.~\ref{fig:tlt_example} corresponding to $\varphi_L$. Then, the only remaining branch corresponds to $\varphi_R$, forcing the vehicle to switch to the right pass for the parking maneuver. Consequently, the backward reachable tube becomes much smaller and the system has stricter constraints on control inputs. Without needing to reconstruct the TLT, it is possible to react to the unplanned block by following the constraints imposed by computing least-restrictive control sets based on $\varphi_R$. Constructing the TLT in this simulation, with a 30 second time horizon, took 65.01 seconds on a system with an AMD Ryzen Threadripper 3970X and an NVIDIA GeForce RTX 2080 Ti while computing the least-restrictive control set took 18 milliseconds on average.

%% file: tlt_hj/src/7_conclusion.tex

In this paper, we present an approach for guaranteeing a CPS completes a LTL task by synthesizing least-restrictive control sets from TLTs constructed from HJ reachability analysis. We are able to take advantage of the richness of LTL task specification together with the strong, low-level guarantees provided by HJ reachability analysis. We detail the key HJ reachability computations required to construct a TLT and provide an algorithm that verifies the existence of control policies that satisfy the constructed TLT. Then, once we know there exists control policies that satisfy the constructed TLT, we can efficiently compute non-empty, least-restrictive control sets that guarantee the CPS completes the specified task. Moreover, to implement and evaluate the approach, we develop a new Python toolbox for constructing and working with TLT. Using this toolbox, we showcase the efficacy of the approach on a simulated driving example where we visualize the evolution of the satisfaction sets and least-restrictive control sets throughout the driving task. Our future work includes both the development of computationally-efficient approaches to directly resolving leaking corner issues and deploying this method on real hardware.

%% file: tlt_hj/src/8_appendix.tex
\section{\texttt{ctrlExists} proofs}\label{app:ctrl_exists}

\begin{proof}
    \textbf{(Lemma~\ref{lemma:not})} Since $\Phi$ and $\Phi'$ are the largest sets satisfying $\varphi$ and $\ltlnot \varphi$, respectively, $\Phi' = \Phi^C$. If $\hat \Phi$ over-approximates $\Phi$, then $\hat \Phi$ contains states that are not in $\Phi$, thus, there must be states in $\Phi^C = \Phi'$ that are not in $\hat \Phi^C$. The same reasoning applies to when $\hat \Phi$ under-approximates $\Phi$.
\end{proof}

\begin{proof}
\textbf{(Lemma~\ref{lemma:and})}
We start by showing that for state $z$,
$$z\in \Phi \Rightarrow z \in \Phi_1 \cap \Phi_2.$$
By Definition~\ref{def:true_sat}, $\forall t \in [0, T]$, if $z\in \Phi$, then
$\exists u(\cdot),\ (\zeta(\cdot; z, 0, u(\cdot)), t) \ltlsatisfy \varphi = \varphi_1 \ltland \varphi_2.$
Thus, we know that $\forall t \in [0, T]$, $(\zeta(\cdot; z, 0, u(\cdot)), t) \ltlsatisfy \varphi_1$ and $(\zeta(\cdot; z, 0, u(\cdot)), t) \ltlsatisfy \varphi_2$. Thus, $z \in \Phi_1$ and $z \in \Phi_2$. Next, we show that
$$z \in \Phi_1 \cap \Phi_2 \not\Rightarrow z \in \Phi.$$
We know that the intersection of $\Phi_1$ and $\Phi_2$ corresponds to a collection of starting states $z$ where, $\forall t\in [0, T]$, $\exists u_1(\cdot)$ where $(\zeta_1(t; z, 0, u_1(\cdot)), t)$ individually satisfies $\varphi_1$ and $\exists u_2(\cdot)$ where $(\zeta_2(t; z, 0, u_2(\cdot)), t)$ individually satisfies $\varphi_2$. However, there can exist $z$ where $u_1(\cdot)$ and $u_2(\cdot)$ result in trajectories that satisfy one of $\varphi_1$ or $\varphi_2$, but not the other.
\end{proof}

\begin{proof}
\textbf{(Lemma~\ref{lemma:or})}
We start by showing that for state $z$, $$z \in \Phi \Rightarrow z \in \Phi_1 \cup \Phi_2 .$$
By Definition~\ref{def:true_sat}, $\forall t \in [0, T]$, if $z\in \Phi$, then
$\exists u(\cdot),\ (\zeta(t; z, 0, u(\cdot)), t) \ltlsatisfy \varphi = \varphi_1 \ltlor \varphi_2.$ Thus, we know that $\forall t \in [0, T]$, $\exists u(\cdot)$ where either $(\zeta(t; z, 0, u(\cdot)), t) \ltlsatisfy \varphi_1$ or $(\zeta(t; z, 0, u(\cdot)), t) \ltlsatisfy \varphi_2$. Thus, $z \in \Phi_1$ or $z \in \Phi_2$. Next, we show that 
$$z \in \Phi_1 \cup \Phi_2 \Rightarrow z \in \Phi.$$
The union of $\Phi_1$ and $\Phi_2$ corresponds to a collection of starting states $z$ where, $\forall t \in [0,T]$, $\exists u_1(\cdot)$ where $(\zeta_1(t; z, 0, u_1(\cdot)), t) \ltlsatisfy \varphi_1$ or $\exists u_2(\cdot)$ where $(\zeta_2(t; z, 0, u_2(\cdot)), t) \ltlsatisfy \varphi_2$. This corresponds to the full set of states that satisfy $\varphi_1 \ltlor \varphi_2$.
\end{proof}

\begin{proof}
\textbf{(Theorem~\ref{thm:approx})}
Since \texttt{ctrlExists} is a recursive algorithm, we start with the bottom of the recursion and prove that all encountered cases when returning up the recursion correctly determines the approximation direction:
    \begin{itemize}
        \item \textbf{Case: leaf node} (Line~\ref{ln:leaf}) This is the base case of the recursion. Leaf nodes are state sets that usually have approx. direction \texttt E, however they can sometimes also have \texttt O or \texttt U.
        \item \textbf{Case: parent is $\ltluntil$ (or $\ltlalways$)} (Line~\ref{ln:until_alw}) Child of $\ltluntil$ (or $\ltlalways$) can have approx. direction of \texttt E, \texttt O, or \texttt U. If child has approx. direction of \texttt E, reachability operator can be either $\oreach$ (or $\orci$) or $\ureach$ (or $\urci$). If child has approx. direction of \texttt O or \texttt U, then reachability operator needs to be corresponding $\oreach$ (or $\orci)$ or $\ureach$ (or $\urci$), respectively. Parent of $\ltluntil$ (or $\ltlalways$) will have approx. direction \texttt O if reachability operator is $\oreach$ (or $\orci$) and \texttt U if reachability operator is $\ureach$ (or $\urci$).
        \item \textbf{Case: parent is $\ltlnot$} (Line~\ref{ln:not}) Child of $\ltlnot$ can have an approx. direction of \texttt E, \texttt O, \texttt U. If child has an approx. direction of \texttt E, parent of $\ltlnot$ has an approx. direction of \texttt E. By Lemma~\ref{lemma:not}, if child has an approx. direction of \texttt O or \texttt U, parent of $\ltlnot$ has approx. direction of \texttt U or \texttt O, respectively. These three cases correspond to multiplying the approx. direction by $-1$.
        \item \textbf{Case: parent is $\ltland$} (Line~\ref{ln:and}) From Lemma~\ref{lemma:and}, we know that the set operation underlying $\ltland$ results in a possible over-approximation of the true satisfaction set. Thus, children of $\ltland$ can have approx. directions of either \texttt E or \texttt O. If either children have an approx. direction of \texttt U, then approx. direction of the parent of $\ltland$ is invalid. If children have approx. direction of either \texttt E or \texttt O, then the parent of $\ltland$ has an approx. direction of \texttt O.
        \item \textbf{Case: parent is $\ltlor$} (Line~\ref{ln:or}) From Lemma~\ref{lemma:or}, we know that the set operation underlying $\ltlor$ is an exact operation. Thus, children of $\ltlor$ can have approx. direction of \texttt E, \texttt U, or \texttt O. If both children have approx. direction of \texttt E, then parent of $\ltlor$ has an approx. direction of \texttt E. If one child has approx. direction of \texttt O or \texttt U, the parent of $\ltlor$ has the same approx. direction.
    \end{itemize}
\end{proof}

\section{Least-Restrictive Control Set Computation Using odp}\label{app:lcrs_comp}
The following is a code example of how to implement Eq.~\ref{eq:lrctrlbound} using OptimizedDp~\cite{Bui2022}, which automatically parallelizes Python code using the heterocl framework~\cite{heterocl}. 

\begin{lstlisting}[language=Python, numbers=left, escapechar=|]
'''
Known variables:
dynSys(System dynamics); grid(Computational grid); initial_value_f(0 sublevel set represents the target); tau(the whole time span); f() and g()(functions in Eq.13).
'''

from odp.solver import HJSolver, computeSpatDerivArray
from odp.compute_trajectory import spa_deriv

def least_restrictive_ctrl( yNext, state, spat_deriv, deltaT):
'''
yNext: V(z, t+deltaT) As the computation is going backward, the value at "t+delta t" is the previous time step of value at "t".
'''
    a = yNext + spat_deriv.T @ f(state) * deltaT
    bT = spat_deriv.T @ g(state) * deltaT
    return (a, bT)

# Numerically compute the value function VPhi
V = HJSolver(dynSys, grid, initial_value_f, tau, saveAllTimeSteps=True)
deltaT = tau[1] - tau[0]
y = V[..., t]
yNext = V[..., t+deltaT]

# Compute Spatial derivatives at every state
index = grid.get_index(dynSys.x)
spat_deriv = spa_deriv(index, y(index), grid)
a, bT = least_restrictive_ctrl(yNext(index), dynSys.x, spat_deriv, deltaT)
\end{lstlisting}